\documentclass[11pt]{article} % For LaTeX2e

\usepackage{amsfonts}
\usepackage{amssymb}
\usepackage{amsthm}
\usepackage{amsmath}
\usepackage{array}
\usepackage{mdwmath}
\usepackage{mdwtab}
\usepackage{eqparbox}
\usepackage{subfig}
\usepackage{stfloats}
\usepackage{fixltx2e}
\usepackage{mathrsfs}
\usepackage{graphicx}
\usepackage{parskip}
\usepackage[numbers]{natbib}
\usepackage[ruled]{algorithm2e}
\usepackage[tmargin=1.0in,bmargin=1.0in,rmargin=1.0in,lmargin=1.0in]{geometry}

\usepackage{setspace}

\theoremstyle{plain}
\newtheorem{thm}{\protect\theoremname}
  \theoremstyle{definition}
  
  \theoremstyle{remark}
  
  \theoremstyle{plain}
  \theoremstyle{definition}
  
\newtheorem{lemma}{Lemma}
\newtheorem{remark}{Remark}
\makeatother
\usepackage[english]{babel}
  
  \providecommand{\claimname}{Claim}
  \providecommand{\definitionname}{Definition}
  \providecommand{\examplename}{Example}
\providecommand{\theoremname}{Theorem}

\newcommand{\norm}[1]{\left\lVert#1\right\rVert}
\newcommand{\abs}[1]{\left\lvert#1\right\rvert}

\newcommand{\sign}{\operatorname{sgn}}
\newcommand{\Tr}{\operatorname{Tr}}
\def\ci{\perp\!\!\!\perp}
\title{Iterative Thresholding Algorithm for Sparse Inverse Covariance Estimation}

{\footnotesize
\author{
Dominique Guillot\footnote{Equal contributors.} \\
Dept. of Statistics \\
Stanford University \\
 Stanford, CA 94305\\
\texttt{dguillot@stanford.edu} \\
\and
Bala Rajaratnam$^\ast$\\
Dept. of Statistics \\
Stanford University \\
 Stanford, CA 94305 \\
\texttt{brajarat@stanford.edu} \\
\and
Benjamin T. Rolfs$^\ast$ \\
ICME\\
Stanford University\\
Stanford, CA 94305 \\
\texttt{benrolfs@stanford.edu} \\
\and
 Arian Maleki \\
Dept. of ECE\\
Rice University \\
 Houston, TX 77005\\
\texttt{arian.maleki@rice.edu} \\
\and
 Ian Wong \\
 Dept. of EE and Statistics\\
 Stanford University\\
 Stanford, CA 94305\\
\texttt{ianw@stanford.edu}
}

}

\date{}
\begin{document}
\maketitle

\begin{abstract}
The $\ell_1$-regularized maximum likelihood estimation problem has recently become a topic of great interest within the machine learning, statistics, and optimization communities as a method for producing sparse inverse covariance estimators. In this paper, a proximal gradient method (\texttt{G-ISTA}) for performing $\ell_1$-regularized covariance matrix estimation is presented. Although numerous algorithms have been proposed for solving this problem, this simple proximal gradient method is found to have attractive theoretical and numerical properties. \texttt{G-ISTA} has a linear rate of convergence, resulting in an $\mathcal{O}(\log \varepsilon)$ iteration complexity to reach a tolerance of $\varepsilon$. This paper gives eigenvalue bounds for the \texttt{G-ISTA} iterates, providing a closed-form linear convergence rate. The rate is shown to be closely related to the condition number of the optimal point. Numerical convergence results and timing comparisons for the proposed method are presented. \texttt{G-ISTA} is shown to perform very well, especially when the optimal point is well-conditioned.  
\end{abstract}

\onehalfspacing

\section{Introduction}
\fontsize{12}{15}
\selectfont

Datasets from a wide range of modern research areas are increasingly high dimensional, which presents a number of theoretical and practical challenges. A fundamental example is the problem of estimating the covariance matrix from a dataset of $n$ samples $\{X^{(i)} \}_{i=1}^n$, drawn \emph{i.i.d} from a $p$-dimensional, zero-mean, Gaussian distribution with covariance matrix $\Sigma \in \mathbb{S}_{++}^p$, $X^{(i)} \sim \mathcal{N}_p(0,\Sigma)$, where $\mathbb{S}_{++}^p$ denotes the space of $p \times p$ symmetric, positive definite matrices. When $n\geq p$ the maximum likelihood covariance estimator $\hat{\Sigma}$ is the sample covariance matrix $S = \frac{1}{n}\sum_{i=1}^n X^{(i)} X^{(i)^T}$. A problem however arises when $n<p$, due to the rank-deficiency in $S$. In this sample deficient case, common throughout several modern applications such as genomics, finance, and earth sciences, the matrix $S$ is not invertible, and thus cannot be directly used to obtain a well-defined estimator for the inverse covariance matrix $\Omega := \Sigma^{-1}$.

A related problem is the inference of a Gaussian graphical model (\cite{WhittakerBook, Lauritzen1996}), that is, a sparsity pattern in the inverse covariance matrix, $\Omega$. Gaussian graphical models provide a powerful means of dimensionality reduction in high-dimensional data. Moreover, such models allow for discovery of conditional independence relations between random variables since, for multivariate Gaussian data, sparsity in the inverse covariance matrix encodes conditional independences. Specifically, if $X = (X_i)_{i=1}^p \in \mathbb{R}^p$ is distributed as $X \sim \mathcal{N}_p(0, \Sigma)$, then $ (\Sigma^{-1})_{ij}= \Omega_{ij} = 0 \iff X_i \ci X_j | \{X_k\}_{k\neq i,j}$, where the notation $A \ci B | C$ denotes the conditional independence of $A$ and $B$ given the set of variables $C$ (see \cite{WhittakerBook, Lauritzen1996}). If a dataset, even one with $n\gg p$ is drawn from a normal distribution with sparse inverse covariance matrix $\Omega$, the inverse sample covariance matrix $S^{-1}$ will almost surely be a dense matrix, although the estimates for those $\Omega_{ij}$ which are equal to 0 may be very small in magnitude. As sparse estimates of $\Omega$ are more robust than $S^{-1}$, and since such sparsity may yield easily interpretable models, there exists significant impetus to perform sparse inverse covariance estimation in very high dimensional low sample size settings.

\citet{Banerjee08} proposed performing such sparse inverse covariance estimation by solving the $\ell_1$-penalized maximum likelihood estimation problem,

\begin{alignat}{1}
\label{eq:sparse_estimator}
\Theta^{\ast}_\rho & = \arg \min_{\Theta \in \mathbb{S}^{p}_{++}} -\log \det \Theta + \langle S, \Theta \rangle + \rho \norm{\Theta}_1,
\end{alignat}

where $\rho>0$ is a penalty parameter, $\langle S, \Theta \rangle = \Tr{(S\Theta)}$, and $\norm{\Theta}_1 = \sum_{i,j}|\Theta_{ij}|$. For $\rho>0$, Problem \eqref{eq:sparse_estimator} is strongly convex and hence has a unique solution, which lies in the positive definite cone $\mathbb{S}^{p}_{++}$ due to the $\log \det$ term, and is hence invertible. Moreover, the $\ell_1$ penalty induces sparsity in $\Theta^{\ast}_\rho$, as it is the closest convex relaxation of the $0-1$ penalty, $\norm{\Theta}_0 = \sum_{i,j} \mathbb{I}(\Theta_{ij} \neq 0)$, where $\mathbb{I}(\cdot)$ is the indicator function \cite{Boyd2004Book}. The unique optimal point of problem \eqref{eq:sparse_estimator}, $\Theta^{\ast}_\rho$, is both invertible (for $\rho > 0$) and sparse (for sufficiently large $\rho$), and can be used as an inverse covariance matrix estimator.

In this paper, a proximal gradient method for solving Problem \eqref{eq:sparse_estimator} is proposed.  The resulting ``graphical iterative shrinkage thresholding algorithm", or \texttt{G-ISTA}, is shown to converge at a linear rate to $\Theta^{\ast}_\rho$, that is, its iterates $\Theta_t$ are proven to satisfy
\begin{alignat}{1}
\norm{\Theta_{t+1} - \Theta^{\ast}_\rho}_F & \leq s \norm{\Theta_t - \Theta^{\ast}_\rho}_F,
\end{alignat}
for a fixed worst-case contraction constant $s \in (0,1)$, where $\norm{\cdot}_F$ denotes the Frobenius norm. The convergence rate $s$ is provided explicitly in terms of $S$ and $\rho$, and importantly, is related to the condition number of $\Theta^{\ast}_\rho$.

We also note that methods outside the penalized likelihood framework have been proposed in the context of graphical models. In particular graphical model estimation and related problems has also be undertaken either in the Bayesian or testing frameworks. The reader is referred to the theoretical work of \citep{Dawid1993, Hero2011, Hero2012, Khare2011,  Letac2007, Rajaratnam2008}, among others, for greater detail. 

The paper is organized as follows. Section \ref{sec:prior_work} describes prior work related to solution of Problem \eqref{eq:sparse_estimator}. The \texttt{G-ISTA} algorithm is formulated in Section \ref{sec:methodology}. Section \ref{sec:convergence} contains the convergence proofs of this algorithm, which constitutes the primary mathematical result of this paper. Numerical results are presented in Section \ref{sec:results}, and concluding remarks are made in Section \ref{sec:conclusion}.
  
\section{Prior Work}
\label{sec:prior_work}

While several excellent general convex solvers exist (for example, \cite{cvx} and \cite{Becker2010}), these are not always adept at handling high dimensional problems (i.e., $p>1000$). As many modern datasets have several thousands of variables, numerous authors have proposed efficient algorithms designed specifically to solve the $\ell_1$-penalized sparse maximum likelihood covariance estimation problem \eqref{eq:sparse_estimator}. 

These can be broadly categorized as either primal or dual methods. Following the literature, we refer to primal methods as those which directly solve Problem \eqref{eq:sparse_estimator}, yielding a concentration estimate. Dual methods \cite{Banerjee08} yield a covariance matrix by solving the constrained problem,
\begin{align}
\label{eq:dual_problem}
\begin{array}{ll}
\underset{U \in \mathbb{R}^{p\times p}}{\mbox{minimize}} & - \log \det ( S + U) - p \\
\mbox{subject to} & \norm{U}_\infty \leq \rho,
\end{array}
\end{align}
where the primal and dual variables are related by $\Theta = (S+U)^{-1}$. Both the primal and dual problems can be solved using block methods (also known as ``row by row" methods), which sequentially optimize one row/column of the argument at each step until convergence. The primal and dual block problems both reduce to $\ell_1$-penalized regressions, which can be solved very efficiently. 

\subsection{Dual Methods}

A number of dual methods for solving Problem \eqref{eq:sparse_estimator} have been proposed in the literature. \citet{Banerjee08} consider a block coordinate descent algorithm to solve the block dual problem, which reduces each optimization step to solving a box-constrained quadratic program. Each of these quadratic programs is equivalent to performing a ``lasso" ($\ell_1$-regularized) regression. \citet{Friedman08} iteratively solve the lasso regression as described in \cite{Banerjee08}, but do so using coordinate-wise descent. Their widely used solver, known as the graphical lasso (\texttt{glasso}) is implemented on \texttt{CRAN}. Global convergence rates of these block coordinate methods are unknown.   \citet{dAspremont2008} use Nesterov's smooth approximation scheme, which produces an $\varepsilon$-optimal solution in $\mathcal{O}(1 / \varepsilon)$ iterations. A variant of Nesterov's smooth method is shown to have a $\mathcal{O}(1 / \sqrt{\varepsilon})$ iteration complexity in \cite{Lu2009, Lu2010}. 

\subsection{Primal Methods}
Interest in primal methods for solving Problem \eqref{eq:sparse_estimator} has been growing for many reasons. One important reason stems from the fact that convergence within a certain tolerance for the dual problem does not necessarily imply convergence within the same tolerance for the primal. 

\citet{YuanLin2007} use interior point methods based on the max-det problem studied in \cite{Vandenberghe96}. \citet{Yuan2009alternating} use an alternating-direction method, while \citet{Scheinberg2010ADMM} proposes a similar method and show a sublinear convergence rate. \citet{MazumderInsights2011} consider block-coordinate descent approaches for the primal problem, similar to the dual approach taken in \cite{Friedman08}. \citet{MazumderPrimal2011} also solve the primal problem with block-coordinate descent, but at each iteration perform a partial as opposed to complete block optimization, resulting in a decreased computational complexity per iteration. Convergence rates of these primal methods have not been considered in the literature and hence theoretical guarantees are not available. \citet{Hsieh2011_QUIC} propose a second-order proximal point algorithm, called \texttt{QUIC}, which converges superlinearly locally around the optimum. 

\section{Methodology}
\label{sec:methodology}

In this section, the \emph{graphical iterative shrinkage thresholding algorithm} (\texttt{G-ISTA}) for solving the primal problem \eqref{eq:sparse_estimator} is presented. A rich body of mathematical and numerical work exists for general iterative shrinkage thresholding and related methods; see, in particular, \cite{Beck2009, Combettes2005, Nesterov1983, NesterovBook, Nesterov2007, Tseng2008Survey}. A brief description is provided here. 

\subsection{General Iterative Shrinkage Thresholding (ISTA)}
\label{subsec:ISTA}

Iterative shrinkage thresholding algorithms (ISTA) are general first-order techniques for solving problems of the form 
\begin{alignat}{1}
\label{eq:split_problem}
\underset{x \in \mathcal{X}}{\mbox{minimize}} & \ F(x) := f(x) + g(x),
\end{alignat} 
where $\mathcal{X}$ is a Hilbert space with inner product $\langle \cdot, \cdot \rangle$ and associated norm $\norm{\cdot}$, $f : \mathcal{X}\rightarrow \mathbb{R}$ is a continuously differentiable, convex function, and $g : \mathcal{X} \rightarrow \mathbb{R}$ is a lower semi-continuous, convex function, not necessarily smooth. The function $f$ is also often assumed to have Lipschitz-continuous gradient $\nabla f$, that is, there exists some constant $L>0$ such that 
\begin{alignat}{1}
\norm{\nabla f(x_1) - \nabla f(x_2)} & \leq L \norm{ x_1 - x_2}
\end{alignat}
for any $x_1, x_2 \in \mathcal{X}$. 

For a given lower semi-continuous convex function $g$, the proximity operator of $g$, denoted by $\mbox{prox}_{g} : \mathcal{X}\rightarrow \mathcal{X}$, is given by
\begin{alignat}{1}
\label{eq:prox_function}
\mbox{prox}_{g}(x) &= \arg\min_{y\in \mathcal{X}} \left\{g(y) +\frac{1}{2}\norm{x-y}^{2}\right\},
\end{alignat}
It is well known (for example, \cite{Combettes2005}) that $x^\ast \in \mathcal{X}$ is an optimal solution of problem \eqref{eq:split_problem} if and only if
\begin{alignat}{1}
x^\ast & = \mbox{prox}_{\zeta g}(x^\ast - \zeta \nabla f(x^\ast))
\end{alignat} 
for any $\zeta > 0$. The above characterization suggests a method for optimizing problem \eqref{eq:split_problem} based on the iteration 
\begin{alignat}{1}
\label{eq:prox_iter}
x_{t+1} = \mbox{prox}_{\zeta_t g} \left( x_t - \zeta_t \nabla f (x_t) \right)
\end{alignat}
for some choice of step size, $\zeta_t$. This simple method is referred to as an iterative shrinkage thresholding algorithm (ISTA). For a step size $\zeta_t \leq \frac{1}{L}$, the ISTA iterates $x_t$ are known to satisfy 
\begin{alignat}{1}
F(x_t) - F(x^\ast) & \simeq \mathcal{O}\left(\frac{1}{t}\right), \forall t,
\end{alignat}
where $x^{\ast}$ is some optimal point, which is to say, they converge to the space of optimal points at a sublinear rate. If no Lipschitz constant $L$ for $\nabla f$ is known, the same convergence result  still holds for $\zeta_t$ chosen such that 
\begin{alignat}{1}
\label{eq:line_search}
f(x_{t+1}) & \leq Q_{\zeta_t}(x_{t+1}, x_t),
\end{alignat}
where $Q_\zeta (\cdot, \cdot) : \mathcal{X} \times \mathcal{X} \rightarrow \mathbb{R}$ is a quadratic approximation to $f$, defined by
\begin{alignat}{1}
\label{eq:quad_approx}
Q_{\zeta} (x,y) & = f(y) + \langle x - y, \nabla f(y) \rangle + \frac{1}{2\zeta}\norm{x- y}^2.
\end{alignat}
See \cite{Beck2009} for more details. 

\subsection{Graphical Iterative Shrinkage Thresholding (\texttt{G-ISTA})}
\label{subsec:G-ISTA}

The general method described in Section \ref{subsec:ISTA} can be adapted to the sparse inverse covariance estimation Problem \eqref{eq:sparse_estimator}. Using the notation introduced in Problem \eqref{eq:split_problem}, define $f,g : \mathbb{S}^p_{++} \rightarrow \mathbb{R}$ by $f(X) = -\log\det(X) + \langle S,X \rangle$ and $g(X) = \rho \norm{X}_1$. Both are continuous convex functions defined on $\mathbb{S}^{p}_{++}$. Although the function $\nabla f(X) = S - X^{-1}$ is not Lipschitz continuous over $\mathbb{S}^p_{++}$, it is Lipschitz continuous within any compact subset of $\mathbb{S}^p_{++}$ (See Lemma \ref{lem:lipschitz_grad} of the Supplemental section).
\vspace{.5pc}
\begin{lemma}[\cite{Banerjee08, Lu2009}]
\label{lem:opt_bound}
The solution of Problem \eqref{eq:sparse_estimator}, $\Theta^{\ast}_\rho$, satisfies $\alpha I \preceq \Theta^{\ast}_\rho \preceq \beta I$, for
\begin{alignat}{2}
\alpha & = \frac{1}{\norm{S}_{2} + p \rho}, \quad \beta & = \min \left\{ \frac{p - \alpha \Tr(S)}{\rho}, \gamma \right\}, 
\end{alignat}
and
\begin{alignat}{1}
\gamma & = \left \{ \begin{array}{ll} \min\{\mathbf{1}^T \abs{S^{-1}}\mathbf{1}, (p - \rho \sqrt{p} \alpha) \norm{S^{-1}}_{2} - (p-1)\alpha \}& \text{if }S\in \mathbb{S}_{++}^p \\
2\mathbf{1}^T\abs{(S + \frac{\rho}{2}I)^{-1}}\mathbf{1} - \Tr((S + \frac{\rho}{2}I)^{-1}) &\text{otherwise,}\end{array}\right.
\end{alignat}
where $I$ denotes the $p\times p$ dimensional identity matrix and $\mathbf{1}$ denotes the $p$-dimensional vector of ones.
\end{lemma}

Note that $f+g$ as defined is a continuous, strongly convex function on $\mathbb{S}^p_{++}$. Moreover, by Lemma \ref{lem:lipschitz_grad} of the supplemental section, $f$ has a Lipschitz continuous gradient when restricted to the compact domain $a I \preceq \Theta \preceq b I$. Hence, $f$ and $g$ as defined meet the conditions described in Section \ref{subsec:ISTA}. 

The proximity operator of $\rho \norm{X}_1$ for $\rho >0$ is the soft-thresholding operator, $\eta_{\rho} : \mathbb{R}^{p\times p} \rightarrow \mathbb{R}^{p\times p}$, defined entrywise by
\begin{alignat}{1}
\left[\eta_\rho (X) \right]_{i,j} = \sign(X_{i,j})\left(|X_{i,j}| - \rho \right)_+, 
\end{alignat}
where for some $x \in \mathbb{R}$, $(x)_+ := \max (x,0)$ (see \cite{Combettes2005}). Finally, the quadratic approximation $Q_{\zeta_t}$ of $f$, as in equation \eqref{eq:quad_approx}, is given by 
\begin{alignat}{1}
\label{eq:Q_ISTA}
Q_{\zeta_t}(\Theta_{t+1}, \Theta_t) & =  -\log \det(\Theta_t) + \langle S, \Theta_t \rangle  + \langle \Theta_{t+1} - \Theta_t, S -\Theta_t^{-1} \rangle +  \frac{1}{2\zeta_{t}}\norm{\Theta_{t+1} - \Theta_t}_{F}^2.  
\end{alignat}

The \texttt{G-ISTA} algorithm for solving Problem \eqref{eq:sparse_estimator} is given in Algorithm \ref{alg:ISTA}. As in \cite{Beck2009}, the algorithm uses a backtracking line search for the choice of step size. The procedure terminates when a pre-specified duality gap is attained. The authors found that an initial estimate of $\Theta_0$ satisfying $[\Theta_0]_{ii} = (S_{ii} + \rho)^{-1}$ works well in practice.  Note also that the positive definite check of $\Theta_{t+1}$ during Step (1) of Algorithm \ref{alg:ISTA} is accomplished using a Cholesky decomposition, and the inverse of $\Theta_{t+1}$ is computed using that Cholesky factor.

\begin{algorithm}[H]
\small
\caption{\texttt{G-ISTA} for Problem \eqref{eq:sparse_estimator}}
\label{alg:ISTA}
\DontPrintSemicolon
\SetKwInOut{Input}{input}\SetKwInOut{Output}{output}
\Input{Sample covariance matrix $S$, penalty parameter $\rho$, tolerance $\varepsilon$, backtracking constant $c \in (0,1)$, initial step size $\zeta_{1,0}$, initial iterate $\Theta_0$. Set $\Delta := 2\varepsilon$.}
\While{ $\Delta > \varepsilon$ } {  
\textit{(1) Line search:} Let $\zeta_{t}$ be the largest element of $\{c^j \zeta_{t,0}\}_{j = 0,1,\ldots}$ so that for $\Theta_{t+1} = \eta_{\zeta_{t} \rho}\left(\Theta_{t} - \zeta_{t}(S - \Theta_{t}^{-1})\right)$, the following are satisfied:
\[
\Theta_{t+1} \succ 0 \quad \text{and} \quad f(\Theta_{t+1}) \leq Q_{\zeta_{t}}(\Theta_{t+1}, \Theta_{t}), 
\]
for $Q_{\zeta_t}$ as defined in \eqref{eq:Q_ISTA}. 
\;
\textit{(2) Update iterate:} $\Theta_{t+1} = \eta_{\zeta_{t} \rho}\left(\Theta_{t} - \zeta_{t}(S - \Theta_{t}^{-1})\right)$ \;
\textit{(3) Set next initial step, $\zeta_{t+1,0}$.} See Section \ref{subsubsec:initial_step}. \;
\textit{(4) Compute duality gap:}
\begin{alignat}{1}
\Delta &= -\log\det(S + U_{t+1}) - p - \log\det\Theta_{t+1} + \langle S, \Theta\rangle + \rho \norm{\Theta_{t+1}}_1 \nonumber, 
\end{alignat}
where $(U_{t+1})_{i,j} = \min\{\max\{([\Theta_{t+1}^{-1}]_{i,j} - S_{i,j}), -\rho\}, \rho\}$.\;
}
\Output{$\varepsilon$-optimal solution to problem \eqref{eq:sparse_estimator}, $\Theta^\ast_\rho = \Theta_{t+1}$.}
\end{algorithm}

\subsubsection{Choice of initial step size, $\zeta_0$}
\label{subsubsec:initial_step}
Each iteration of Algorithm \ref{alg:ISTA} requires an initial step size, $\zeta_0$. The results of Section \ref{sec:convergence} guarantee that any $\zeta_0 \leq \lambda_{\min}(\Theta_{t})^2$ will be accepted by the line search criteria of Step 1 in the next iteration. However, in practice this choice of step is overly cautious; a much larger step can often be taken. Our implementation of Algorithm \ref{alg:ISTA} chooses the Barzilai-Borwein step \cite{Barzilai1988}. This step, given by
\begin{align}
\zeta_{t+1,0} = \frac{\Tr{((\Theta_{t+1} - \Theta_t)(\Theta_{t+1} - \Theta_t))}}{\Tr{((\Theta_{t+1} - \Theta_t)(\Theta_{t}^{-1} - \Theta_{t+1}^{-1}))}},
\end{align}
is also used in the SpaRSA algorithm \cite{Wright2009}, and approximates the Hessian around $\Theta_{t+1}$.   If a certain number of maximum backtracks do not result in an accepted step, \texttt{G-ISTA} takes the safe step, $\lambda_{\min}(\Theta_{t})^2$. Such a safe step can be obtained from $\lambda_{\max}(\Theta_{t}^{-1})$, which in turn can be quickly approximated using power iteration. 

\section{Convergence Analysis}
\label{sec:convergence}
In this section, linear convergence of Algorithm \ref{alg:ISTA} is discussed. Throughout the section, $\Theta_t$ $(t = 1, 2, \dots)$ denote the iterates of Algorithm \ref{alg:ISTA}, and $\Theta^{\ast}_\rho$ the optimal solution to Problem \eqref{eq:sparse_estimator} for $\rho>0$. The minimum and maximum eigenvalues of a symmetric matrix $A$ are denoted by $\lambda_{\mbox{min}}(A)$ and $\lambda_{\mbox{\textrm{max}}}(A)$, respectively. 
\vspace{.5pc}
\begin{thm}
\label{thm:ISTA_convergence}
Assume that the iterates $\Theta_t$ of Algorithm \ref{alg:ISTA} satisfy $a I\preceq \Theta_t \preceq b I, \forall t$ for some fixed constants $0< a < b$. If $\zeta_t \leq a^2, \forall t$, then
\begin{align}
\norm{\Theta_{t+1} - \Theta^\ast_\rho}_F & \leq \max \left\{ \abs{1-\frac{\zeta_t}{b^2}},\abs{1-\frac{\zeta_t}{a^2}} \right\} \norm{\Theta_t - \Theta^\ast_\rho}_F.
\end{align}
Furthermore,
\begin{enumerate}
\item The step size $\zeta_t$  which yields an optimal worst-case contraction bound $s(\zeta_t)$ is $\zeta = \frac{2}{a^{-2} + b^{-2}}$. 
\item The optimal worst-case contraction bound corresponding to $\zeta = \frac{2}{a^{-2} + b^{-2}}$ is given by
\begin{align*}
s(\zeta) :&=  1 - \frac{2}{1+\frac{b^2}{a^2}} 
\end{align*}
\end{enumerate}
\end{thm}
\begin{proof}
A direct proof is given in the appendix. Note that linear convergence of proximal gradient methods for strongly convex objective functions in general has already been proven (see Supplemental section). 
\end{proof}

It remains to show that there exist constants $a$ and $b$ which bound the eigenvalues of $\Theta_{t}, \forall t$. The existence of such constants follows directly from Theorem \ref{thm:ISTA_convergence}, as $\Theta_t$ lie in the bounded domain $\{\Theta \in \mathbb{S}_{++}^p \ : \ f(\Theta) + g(\Theta)  < f(\Theta_0) + g(\Theta_0) \}$, for all $t$. However, it is possible to specify the constants $a$ and $b$ to yield an explicit rate; this is done in Theorem \ref{thm:eigenvalue_bounds}. 

\vspace{.5pc}
\begin{thm}
\label{thm:eigenvalue_bounds}
Let $\rho > 0$, define $\alpha$ and $\beta$ as in Lemma \ref{lem:opt_bound}, and assume $\zeta_t \leq \alpha^2, \forall t$. Then the iterates $\Theta_t$ of Algorithm \ref{alg:ISTA} satisfy $\alpha I \preceq \Theta_t \preceq b^\prime I, \forall t$, with $b^\prime = \norm{\Theta^{\ast}_\rho}_2 + \norm{\Theta_0 - \Theta^\ast_\rho}_F \leq \beta + \sqrt{p}(\beta + \alpha)$. 
\end{thm}
\begin{proof}
See the Supplementary section. 
\end{proof}

Importantly, note that the bounds of Theorem \ref{thm:eigenvalue_bounds} depend explicitly on the bound of $\Theta_\rho^\ast$, as given by Lemma \ref{lem:opt_bound}. These eigenvalue bounds on $\Theta_{t+1}$, along with Theorem \ref{thm:ISTA_convergence}, provide a closed form linear convergence rate for Algorithm \ref{alg:ISTA}. This rate depends only on properties of the solution. 
\vspace{.5pc}
\begin{thm}
\label{thm:lin_convergence}
Let $\alpha$ and $\beta$ be as in Lemma \ref{lem:opt_bound}. Then for a constant step size $\zeta_t := \zeta < \alpha^2$, the iterates of Algorithm \ref{alg:ISTA} converge linearly with a rate of 
\begin{align}
\label{eq:shrinkage_constant}
s(\zeta) = 1 - \frac{2\alpha^2}{\alpha^2 + (\beta + \sqrt{p} (\beta - \alpha))^2} <1
\end{align}
\end{thm}
\begin{proof}
By Theorem \ref{thm:eigenvalue_bounds}, for $\zeta < \alpha^2$, the iterates $\Theta_{t}$ satisfy
\[
\alpha I \preceq \Theta_t \preceq \left( \norm{\Theta^\ast_\rho}_2 + \norm{\Theta_0 - \Theta^\ast_\rho}_F\right) I
\]
for all $t$. Moreover, since $\alpha I \preceq \Theta^{\ast} \preceq \beta I$, if $\alpha I \preceq \Theta_0 \preceq \beta I$ (for instance, by taking $\Theta_0 = (S+\rho I)^{-1}$ or some multiple of the identity) then this can be bounded as:
\begin{alignat}{1}
\norm{\Theta^\ast_\rho}_2 + \norm{\Theta_0 - \Theta^\ast_\rho}_F & \leq \beta + \sqrt{p} \norm{\Theta_0 - \Theta^\ast_\rho}_2 \\
& \leq \beta + \sqrt{p} (\beta - \alpha).
\end{alignat}
Therefore, 
\begin{alignat}{1}
\alpha I & \preceq \Theta_t \preceq \left( \beta + \sqrt{p} (\beta - \alpha) \right) I,
\end{alignat}
and the result follows from Theorem \ref{thm:ISTA_convergence}. 
\end{proof}

\medskip

\begin{remark}
Note that the contraction constant (equation \ref{eq:shrinkage_constant}) of Theorem \ref{thm:lin_convergence} is closely related to the condition number of $\Theta^{\ast}_\rho$,
\[
\kappa(\Theta^{\ast}_\rho) = \frac{\lambda_{\mbox{max}}(\Theta^{\ast}_\rho)}{\lambda_{\mbox{min}}(\Theta^{\ast}_\rho)} \leq \frac{\beta}{\alpha}
\]
as 
\begin{alignat}{1}
1 - \frac{2\alpha^2}{\alpha^2 + (\beta + \sqrt{p} (\beta - \alpha))^2} & \geq 1 - \frac{2\alpha^2}{\alpha^2  + \beta^2} \geq 1 - 2\kappa(\Theta^{\ast}_\rho)^{-2}.
\end{alignat}
Therefore, the worst case bound becomes close to 1 as the conditioning number of $\Theta^{\ast}_\rho$ increases. 
\end{remark}

It is important to compare the above convergence results to those that have been recently established. In particular, the useful, recent \texttt{QUIC} method \cite{Hsieh2011_QUIC} warrants a discussion. As soon as the sign pattern of its iterates match that of the true optimum, the non-smooth problem becomes effectively smooth and the \texttt{QUIC} algorithm reduces to a Newton method. At this point, \texttt{QUIC} converges quadratically; however, this is a very local property, and no overall complexity bounds have been specified for \texttt{QUIC}. This can be contrasted with our results, which take advantage of existing bounds on the optimal solution to yield global convergence (i.e., that we can always specify a starting point which meets our conditions). We also note that convergence rates have not been established for the \texttt{glasso}.  However, \texttt{QUIC} and \texttt{glasso} can all be very fast in appropriate settings. Each brings a useful addition to the literature by taking advantage of different structural elements (block structure for \texttt{glasso}, second order approaches for \texttt{QUIC}, and conditioning bounds for \texttt{G-ISTA}). We feel there is no silver bullet; each method outperforms the others in certain settings.

\section{Numerical Results}
\label{sec:results}

In this section, we provide numerical results for the \texttt{G-ISTA} algorithm. In Section \ref{subsec:conditioning}, the theoretical results of Section \ref{sec:convergence} are demonstrated. Section \ref{subsec:timing} compares running times of the \texttt{G-ISTA}, \texttt{glasso} \cite{Friedman08}, and \texttt{QUIC} \cite{Hsieh2011_QUIC} algorithms. All algorithms were implemented in C++, and run on an Intel $i7-2600k$ $3.40\mbox{GHz} \times 8$ core with $16$ GB of RAM. 

\subsection{Synthetic Datasets}
\label{subsec:synthetic_dataset}
Synthetic data for this section was generated following the method used by \cite{Lu2010, MazumderPrimal2011}. For a fixed $p$, a $p$ dimensional inverse covariance matrix $\Omega$ was generated with off-diagonal entries drawn $\emph{i.i.d}$ from a $\mbox{uniform}(-1,1)$ distribution. These entries were set to zero with some fixed probability (in this case, either $0.97$ or $0.85$ to simulate a very sparse and a somewhat sparse model). Finally, a multiple of the identity was added to the resulting matrix so that the smallest eigenvalue was equal to $1$. In this way, $\Omega$ was insured to be sparse, positive definite, and well-conditioned. Datsets of $n$ samples were then generated by drawing \emph{i.i.d.} samples from a $\mathcal{N}_p(0, \Omega^{-1})$ distribution. For each value of $p$ and sparsity level of $\Omega$, $n = 1.2 p$ and $n = 0.2p$ were tested, to represent both the $n <p$ and $n>p$ cases. 

{
\begin{table}[h]
\small
\begin{center}
\begin{tabular}{|c|c|c|c|c|c|}
\cline{2-6} 
\multicolumn{1}{c|}{} & $\rho$ & $0.03$ & $0.06$ & $0.09$ & $0.12$\tabularnewline
\hline 
\textbf{problem} & \textbf{algorithm} & \textbf{time/iter} & \textbf{time/iter} & \textbf{time/iter} & \textbf{time/iter}\tabularnewline
\hline 
\hline 
 & $\mbox{nnz}(\Omega_{\rho}^{\ast})$/$\kappa(\Omega_{\rho}^{\ast})$ & $27.65\%/48.14$ & $15.08\%/20.14$ & $7.24\%/7.25$ & $2.39\%/2.32$\tabularnewline
\cline{2-6} 
$p=2000$ & \texttt{glasso} & $1977.92/11$ & $831.69/8$ & $604.42/7$ & $401.59/5$\tabularnewline
\cline{2-6} 
$n=400$ & \texttt{QUIC} & $1481.80/21$ & $257.97/11$ & $68.49/8$ & $15.25/6$\tabularnewline
\cline{2-6} 
$\mbox{nnz}(\Omega)=3\%$ & \texttt{G-ISTA} & $\mathbf{145.60}/437$ & $\mathbf{27.05}/9$ & $\mathbf{8.05}/27$ & $\mathbf{3.19}/12$\tabularnewline
\hline 
 & $\mbox{nnz}(\Omega_{\rho}^{\ast})$/$\kappa(\Omega_{\rho}^{\ast})$ & $14.56\%/10.25$ & $3.11\%/2.82$ & $0.91\%/1.51$ & $0.11\%/1.18$\tabularnewline
\cline{2-6} 
$p=2000$ & \texttt{glasso} & $667.29/7$ & $490.90/6$ & $318.24/4$ & $233.94/3$\tabularnewline
\cline{2-6} 
$n=2400$ & \texttt{QUIC} & $211.29/10$ & $24.98/7$ & $5.16/5$ & $\mathbf{1.56}/4$\tabularnewline
\cline{2-6} 
$\mbox{nnz}(\Omega)=3\%$ & \texttt{G-ISTA} & $\mathbf{14.09}/47$ & $\mathbf{3.51}/13$ & $\mathbf{2.72}/10$ & $2.20/8$\tabularnewline
\hline 
 & $\mbox{nnz}(\Omega_{\rho}^{\ast})$/$\kappa(\Omega_{\rho}^{\ast})$ & $27.35\%/64.22$ & $15.20\%/28.50$ & $7.87\%/11.88$ & $2.94\%/2.87$\tabularnewline
\cline{2-6} 
$p=2000$ & \texttt{glasso} & $2163.33/11$ & $862.39/8$ & $616.81/7$ & $48.47/7$\tabularnewline
\cline{2-6} 
$n=400$ & \texttt{QUIC} & $1496.98/21$ & $318.57/12$ & $96.25/9$ & $23.62/7$\tabularnewline
\cline{2-6} 
$\mbox{nnz}(\Omega)=15\%$ & \texttt{G-ISTA} & $\mathbf{251.51}/714$ & $\mathbf{47.35}/148$ & $\mathbf{7.96}/28$ & $\mathbf{3.18}/12$\tabularnewline
\hline 
 & $\mbox{nnz}(\Omega_{\rho}^{\ast})$/$\kappa(\Omega_{\rho}^{\ast})$ & $19.98\%/17.72$ & $5.49\%/4.03$ & $65.47\%/1.36$ & $0.03\%/1.09$\tabularnewline
\cline{2-6} 
$p=2000$ & \texttt{glasso} & $708.15/6$ & $507.04/6$ & $313.88/4$ & $233.16/3$\tabularnewline
\cline{2-6} 
$n=2400$ & \texttt{QUIC} & $301.35/10$ & $491.54/17$ & $4.12/5$ & $1.34/4$\tabularnewline
\cline{2-6} 
$\mbox{nnz}(\Omega)=15\%$ & \texttt{G-ISTA} & $\mathbf{28.23}/88$ & $\mathbf{4.08}/16$ & $\mathbf{1.95}/7$ & $\mathbf{1.13}/4$\tabularnewline
\hline 
\end{tabular}
\end{center}
\caption{Timing comparisons for $p=2000$ dimensional datasets, generated as in Section \ref{subsec:synthetic_dataset}. Above, $\mbox{nnz}(A)$ is the percentage of nonzero elements of matrix $A$. }
\label{fig:timing_synth}
\end{table}
}

\subsection{Demonstration of Convergence Rates}
\label{subsec:conditioning}

The linear convergence rate derived for \texttt{G-ISTA} in Section \ref{sec:convergence} was shown to be heavily dependent on the conditioning of the final estimator. To demonstrate these results, \texttt{G-ISTA}  was run on a synthetic dataset, as described in Section \ref{subsec:synthetic_dataset}, with $p = 500$ and $n = 300$. Regularization parameters of  $\rho = 0.75$, 0.1, 0.125, 0.15, and 0.175 were used. Note as $\rho$ increases, $\Theta_\rho^{\ast}$ generally becomes better conditioned. For each value of $\rho$, the numerical optimum was computed to a duality gap of $10^{-10}$ using \texttt{G-ISTA}. These values of $\rho$ resulted in sparsity levels of $81.80\%$, $89.67\%$, $94.97\%$, $97.82\%$, and $99.11\%$, respectively. \texttt{G-ISTA} was then run again, and the Frobenius norm argument errors at each iteration were stored.  These errors were plotted on a log scale for each value of $\rho$ to demonstrate the dependence of the convergence rate on condition number. See Figure \ref{fig:conditioning}, which clearly demonstrates the effects of conditioning.

\begin{figure}[h]
\begin{center}
\includegraphics[width=5.0in]{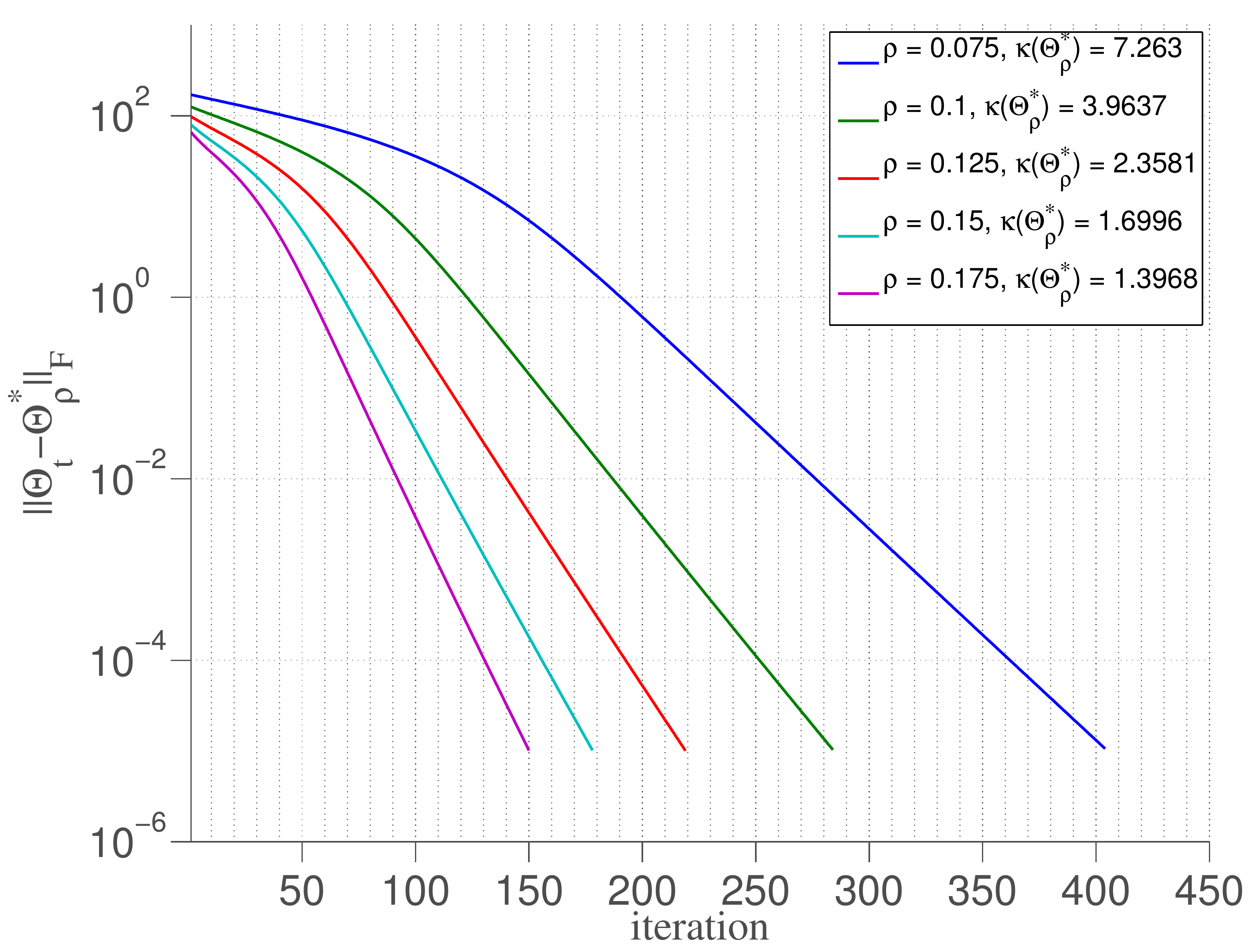}
\end{center}
\caption{Semilog plot of $\norm{\Theta_t - \Theta^{\ast}_\rho}_F$ vs. iteration number $t$, demonstrating linear convergence rates of \texttt{G-ISTA}, and dependence of those rates on $\kappa(\Theta^{\ast}_\rho)$.}
\label{fig:conditioning}
\end{figure}

\subsection{Timing Comparisons}
\label{subsec:timing}

The \texttt{G-ISTA}, \texttt{glasso}, and \texttt{QUIC} algorithms were run on synthetic datasets (real datasets are presented in the Supplemental section) of varying $p,n$ and with different levels of regularization, $\rho$. All algorithms were run to ensure a fixed duality gap, here taken to be $10^{-5}$. This comparison used efficient C++ implementations of each of the three algorithms investigated. The implementation of \texttt{G-ISTA} was adapted from the publicly available C++ implementation of \texttt{QUIC} \citet{Hsieh2011_QUIC}. Running times were recorded and are presented in Table \ref{fig:timing_synth}. Further comparisons are presented in the Supplementary section. 

\medskip

\begin{remark}
The three algorithms variable ability to take advantage of multiple processors is an important detail. The times presented in Table \ref{fig:timing_synth} are wall times, not CPU times. The comparisons were run on a multicore processor, and it is important to note that the Cholesky decompositions and inversions required by both \texttt{G-ISTA} and \texttt{QUIC} take advantage of multiple cores. On the other hand, the $p^2$ dimensional lasso solve of \texttt{QUIC} and $p$-dimensional lasso solve of \texttt{glasso} do not. For this reason, and because Cholesky factorizations and inversions make up the bulk of the computation required by \texttt{G-ISTA}, the CPU time of \texttt{G-ISTA} was typically greater than its wall time by a factor of roughly 4. The CPU and wall times of \texttt{QUIC} were more similar; the same applies to \texttt{glasso}.
\end{remark}

\section{Conclusion}
\label{sec:conclusion}

In this paper, a proximal gradient method was applied to the sparse inverse covariance problem. Linear convergence was discussed, with a fixed closed-form rate. Numerical results have also been presented, comparing \texttt{G-ISTA} to the widely-used \texttt{glasso} algorithm and the newer, but very fast, \texttt{QUIC} algorithm. These results indicate that \texttt{G-ISTA} is competitive, in particular for values of $\rho$ which yield sparse, well-conditioned estimators. The \texttt{G-ISTA} algorithm was very fast on the synthetic examples of Section \ref{subsec:timing}, which were generated from well-conditioned models. For poorly conditioned models, \texttt{QUIC} is very competitive. The Supplemental section gives two real datasets which demonstrate this. For many practical applications however, obtaining an estimator that is well-conditioned is important (\cite{Rolfs12, Won2012Condition}). To conclude, although second-order methods for the sparse inverse covariance estimation problem have recently been shown to perform well, simple first-order methods cannot be ruled out, as they can also be very competitive in many cases. 

\medskip

{\it Acknowledgments:} Dominique Guillot was supported in part by the National Science Foundation under Grant No. DMS-1106642. Bala Rajaratnam was supported in part by the National Science Foundation under Grant Nos. DMS-0906392 (ARRA), DMS-CMG-1025465, AGS-1003823, DMS-1106642 and grants NSA H98230-11-1-0194, DARPA-YFA N66001-11-1-4131, and SUWIEVP10-SUFSC10-SMSCVISG0906. Benjamin Rolfs was supported in part by the Department of Energy Office of Science Graduate Fellowship Program DE-AC05-06OR23100 (ARRA) and NSF grant AGS-1003823. 

\newpage
\bibliographystyle{plainnat}
\bibliography{ista_refs}

\begin{thebibliography}{37}
\providecommand{\natexlab}[1]{#1}
\providecommand{\url}[1]{\texttt{#1}}
\expandafter\ifx\csname urlstyle\endcsname\relax
  \providecommand{\doi}[1]{doi: #1}\else
  \providecommand{\doi}{doi: \begingroup \urlstyle{rm}\Url}\fi

\bibitem[{Banerjee} et~al.(2008){Banerjee}, {El Ghaoui}, and
  d'Aspremont]{Banerjee08}
O.~{Banerjee}, L.~{El Ghaoui}, and A.~d'Aspremont.
\newblock Model selection through sparse maximum likelihood estimation for
  multivarate gaussian or binary data.
\newblock \emph{Journal of Machine Learning Research}, 9:\penalty0 485--516,
  2008.

\bibitem[Barzilai and Borwein(1988)]{Barzilai1988}
Jonathan Barzilai and Jonathan~M. Borwein.
\newblock {Two-Point Step Size Gradient Methods}.
\newblock \emph{IMA Journal of Numerical Analysis}, 8\penalty0 (1):\penalty0
  141--148, 1988.

\bibitem[Beck and Teboulle(2009)]{Beck2009}
Amir Beck and Marc Teboulle.
\newblock A fast iterative shrinkage-thresholding algorithm for linear inverse
  problems.
\newblock \emph{SIAM Journal on Imaging Sciences}, 2:\penalty0 183--202, 2009.
\newblock ISSN 1936-4954.

\bibitem[Becker et~al.(2010)Becker, Candes, and Grant]{Becker2010}
S.~Becker, E.J. Candes, and M.~Grant.
\newblock Templates for convex cone problems with applications to sparse signal
  recovery.
\newblock \emph{Mathematical Programming Computation}, 3:\penalty0 165--218,
  2010.

\bibitem[Boyd and Vandenberghe(2004)]{Boyd2004Book}
Stephen Boyd and Lieven Vandenberghe.
\newblock \emph{Convex Optimization}.
\newblock Cambridge University Press, 2004.

\bibitem[Brohan et~al.(2006)Brohan, Kennedy, Harris, Tett, and Jones]{Brohan06}
P.~Brohan, J.~J. Kennedy, I.~Harris, S.~F.~B. Tett, and P.~D. Jones.
\newblock Uncertainty estimates in regional and global observed temperature
  changes: A new data set from 1850.
\newblock \emph{Journal of Geophysical Research}, 111, 2006.

\bibitem[Chen and Rockafellar(1997)]{Chen1997}
George~H.G. Chen and R.T. Rockafellar.
\newblock Convergence rates in forward-backward splitting.
\newblock \emph{Siam Journal on Optimization}, 7:\penalty0 421--444, 1997.

\bibitem[Combettes and Wajs(2005)]{Combettes2005}
Patrick~L. Combettes and Val\'{e}rie~R. Wajs.
\newblock Signal recovery by proximal forward-backward splitting.
\newblock \emph{Multiscale Modeling \& Simulation}, 4\penalty0 (4):\penalty0
  1168--1200, 2005.

\bibitem[D'Aspremont et~al.(2008)D'Aspremont, Banerjee, and
  El~Ghaoui]{dAspremont2008}
Alexandre D'Aspremont, Onureena Banerjee, and Laurent El~Ghaoui.
\newblock First-order methods for sparse covariance selection.
\newblock \emph{SIAM Journal on Matrix Analysis and Applications}, 30\penalty0
  (1):\penalty0 56--66, 2008.

\bibitem[Dawid and Lauritzen(1993)]{Dawid1993}
A.P. Dawid and S.L. Lauritzen.
\newblock Hyper-markov laws in the statistical analysis of decomposable
  graphical models.
\newblock \emph{Annals of Statistics}, 21:\penalty0 1272--1317, 1993.

\bibitem[{Friedman} et~al.(2008){Friedman}, {Hastie}, and
  {Tibshirani}]{Friedman08}
J.~{Friedman}, T.~{Hastie}, and R.~{Tibshirani}.
\newblock Sparse inverse covariance estimation with the graphical lasso.
\newblock \emph{Biostatistics}, 9:\penalty0 432--441, 2008.

\bibitem[Grant and Boyd(2011)]{cvx}
M.~Grant and S.~Boyd.
\newblock {CVX}: Matlab software for disciplined convex programming, version
  1.21.
\newblock {http://cvxr.com/cvx}, April 2011.

\bibitem[Hero and Rajaratnam(2011)]{Hero2011}
A.~Hero and B.~Rajaratnam.
\newblock Large-scale correlation screening.
\newblock \emph{Journal of the American Statistical Association}, 106\penalty0
  (496):\penalty0 1540--1552, 2011.

\bibitem[Hero and Rajaratnam(2012)]{Hero2012}
A.~Hero and B.~Rajaratnam.
\newblock Hub discovery in partial correlation graphs.
\newblock \emph{IEEE Transactions on Information Theory}, 58\penalty0
  (9):\penalty0 6064--6078, 2012.

\bibitem[Horn and Johnson(1990)]{HornAndJohnson}
Roger~A. Horn and Charles~R. Johnson.
\newblock \emph{Matrix Analysis}.
\newblock Cambridge University Press, 1990.

\bibitem[Hsieh et~al.(2011)Hsieh, Sustik, Dhillon, and
  Ravikumar]{Hsieh2011_QUIC}
Cho-Jui Hsieh, Matyas~A. Sustik, Inderjit~S. Dhillon, and Pradeep~K. Ravikumar.
\newblock Sparse inverse covariance matrix estimation using quadratic
  approximation.
\newblock In \emph{Advances in Neural Information Processing Systems 24}, pages
  2330--2338. 2011.

\bibitem[Khare and Rajaratnam(2011)]{Khare2011}
K.~Khare and B.~Rajaratnam.
\newblock Wishart distributions for decomposable covariance graph models.
\newblock \emph{Annals of Statistics}, 39\penalty0 (1):\penalty0 514--555,
  2011.

\bibitem[Lauritzen(1996)]{Lauritzen1996}
S.L. Lauritzen.
\newblock \emph{Graphical models}.
\newblock Oxford Science Publications. Clarendon Press, 1996.

\bibitem[Letac and Massam(2007)]{Letac2007}
G.~Letac and H.~Massam.
\newblock Wishart distributions for decomposable graphs.
\newblock \emph{Annals of Statistics}, 35\penalty0 (3):\penalty0 1278--1323,
  2007.

\bibitem[Lu(2009)]{Lu2009}
Zhaosong Lu.
\newblock Smooth optimization approach for sparse covariance selection.
\newblock \emph{SIAM Journal on Optimization}, 19\penalty0 (4):\penalty0
  1807--1827, 2009.
\newblock ISSN 1052-6234.
\newblock \doi{http://dx.doi.org/10.1137/070695915}.

\bibitem[Lu(2010)]{Lu2010}
Zhaosong Lu.
\newblock Adaptive first-order methods for general sparse inverse covariance
  selection.
\newblock \emph{SIAM Journal on Matrix Analysis and Applications}, 31:\penalty0
  2000--2016, 2010.

\bibitem[Mazumder and Agarwal(2011)]{MazumderPrimal2011}
Rahul Mazumder and Deepak~K. Agarwal.
\newblock A flexible, scalable and efficient algorithmic framework for the
  \emph{Primal} graphical lasso.
\newblock \emph{Pre-print}, 2011.

\bibitem[Mazumder and Hastie(2011)]{MazumderInsights2011}
Rahul Mazumder and Trevor Hastie.
\newblock The graphical lasso: New insights and alternatives.
\newblock \emph{Pre-print}, 2011.

\bibitem[Nesterov(1983)]{Nesterov1983}
Yurii Nesterov.
\newblock A method of solving a convex programming problem with convergence
  rate ${O}(1/k^2)$.
\newblock \emph{Soviet Mathematics Doklady}, 27(2):\penalty0 372--376, 1983.

\bibitem[Nesterov(2004)]{NesterovBook}
Yurii Nesterov.
\newblock \emph{Introductory Lectures on Convex Optimization}.
\newblock Kluwer Academic Publishers, 2004.

\bibitem[Nesterov(2007)]{Nesterov2007}
Yurii Nesterov.
\newblock Gradient methods for minimizing composite objective function.
\newblock {CORE} discussion papers, Universit\'e catholique de Louvain, Center
  for Operations Research and Econometrics (CORE), 2007.

\bibitem[Pittman et~al.(2004)Pittman, Huang, Dressman, Horng, Cheng, Tsou,
  Chen, Bild, Iversen, Huang, Nevins, and West]{Pittman2004}
Jennifer Pittman, Erich Huang, Holly Dressman, Cheng-Fang~F. Horng, Skye~H.
  Cheng, Mei-Hua~H. Tsou, Chii-Ming~M. Chen, Andrea Bild, Edwin~S. Iversen,
  Andrew~T. Huang, Joseph~R. Nevins, and Mike West.
\newblock {Integrated modeling of clinical and gene expression information for
  personalized prediction of disease outcomes.}
\newblock \emph{Proceedings of the National Academy of Sciences of the United
  States of America}, 101\penalty0 (22):\penalty0 8431--8436, 2004.

\bibitem[Rajaratnam et~al.(2008)Rajaratnam, Massam, and
  Carvalho]{Rajaratnam2008}
B.~Rajaratnam, H.~Massam, and C.~Carvalho.
\newblock Flexible covariance estimation in graphical models.
\newblock \emph{Annals of Statistics}, 36:\penalty0 2818--2849, 2008.

\bibitem[Rolfs and Rajaratnam(2012)]{Rolfs12}
Benjamin~T. Rolfs and Bala Rajaratnam.
\newblock A note on the lack of symmetry in the graphical lasso.
\newblock \emph{Computational Statistics and Data Analysis}, 2012.

\bibitem[Scheinberg et~al.(2010)Scheinberg, Ma, and
  Goldfarb]{Scheinberg2010ADMM}
Katya Scheinberg, Shiqian Ma, and Donald Goldfarb.
\newblock Sparse inverse covariance selection via alternating linearization
  methods.
\newblock In \emph{Advances in Neural Information Processing Systems 23}, pages
  2101--2109. 2010.

\bibitem[Tseng(2008)]{Tseng2008Survey}
Paul Tseng.
\newblock {On accelerated proximal gradient methods for convex-concave
  optimization}.
\newblock \emph{submitted to SIAM Journal on Optimization}, 2008.

\bibitem[Vandenberghe et~al.(1996)Vandenberghe, Boyd, and Wu]{Vandenberghe96}
Lieven Vandenberghe, Stephen Boyd, and Shao-Po Wu.
\newblock Determinant maximization with linear matrix inequality constraints.
\newblock \emph{SIAM Journal on Matrix Analysis and Applications}, 19:\penalty0
  499--533, 1996.

\bibitem[Whittaker(1990)]{WhittakerBook}
J.~Whittaker.
\newblock \emph{{Graphical Models in Applied Multivariate Statistics}}.
\newblock Wiley, 1990.

\bibitem[Won et~al.(2012)Won, Lim, Kim, and Rajaratnam]{Won2012Condition}
J.~Won, J.~Lim, S.~Kim, and B.~Rajaratnam.
\newblock Condition number regularized covariance estimation.
\newblock \emph{Journal of the Royal Statistical Society Series B}, 2012.

\bibitem[Wright et~al.(2009)Wright, Nowak, and Figueiredo]{Wright2009}
Stephen~J. Wright, Robert~D. Nowak, and M\'{a}rio A.~T. Figueiredo.
\newblock Sparse reconstruction by separable approximation.
\newblock \emph{IEE Transactions on Signal Processing}, 57\penalty0
  (7):\penalty0 2479--2493, 2009.

\bibitem[Yuan and Lin(2007)]{YuanLin2007}
Ming Yuan and Yi~Lin.
\newblock Model selection and estimation in the gaussian graphical model.
\newblock \emph{Biometrika}, 94\penalty0 (1):\penalty0 19--35, 2007.

\bibitem[Yuan(2010)]{Yuan2009alternating}
X.M. Yuan.
\newblock Alternating direction method of multipliers for covariance selection
  models.
\newblock \emph{Journal of Scientific Computing}, pages 1--13, 2010.

\end{thebibliography}

\newpage
\newpage
\appendix
\section{Supplementary material}
\subsection{Lipschitz Continuity of $\nabla f(X)$}
\begin{lemma}
\label{lem:lipschitz_grad}
For any $X,Y \in \mathbb{S}_{++}^{p}$,
\[
\frac{1}{b^2} \norm{X-Y}_2 \leq \norm{X^{-1}-Y^{-1}}_2 \leq \frac{1}{a^2} \norm{X - Y}_2,
\]
where $a= \min\{\lambda_{\min}(X), \lambda_{\min}(Y)\}$ and $b = \max\{\lambda_{\max}(X), \lambda_{\max}(Y)\}$. 
\end{lemma}
\begin{proof}
To prove the right-hand side inequality, notice that 
\[
X^{-1}-Y^{-1} = X^{-1}(Y-X)Y^{-1}. 
\]
Thus, 
\begin{eqnarray*}
\norm{X^{-1}-Y^{-1}}_2 &=& \norm{X^{-1}(Y-X)Y^{-1}}_2 \\
&\leq&  \norm{X^{-1}}_2 \norm{X-Y}_2 \norm{Y^{-1}}_2 \\
&=& \lambda_{\textrm{max}}(X^{-1})  \lambda_{\textrm{max}}(Y^{-1}) \norm{X-Y}_2 \\
&=& \frac{1}{\lambda_{\textrm{min}}(X)} \frac{1}{\lambda_{\textrm{min}}(Y)} \norm{X-Y}_2 \\
&\leq& \frac{1}{a^2} \norm{X-Y}_2. 
\end{eqnarray*}
To prove the left inequality, note first that
\[
Y-X = X(X^{-1}-Y^{-1})Y. 
\]
Therefore, 
\begin{eqnarray*}
\norm{X-Y}_2 &=& \norm{X(X^{-1}-Y^{-1})Y}_2 \\
&\leq&  \norm{X}_2 \norm{X^{-1}-Y^{-1}}_2 \norm{Y}_2 \\
&=& \lambda_{\textrm{max}}(X)  \lambda_{\textrm{max}}(Y) \norm{X^{-1}-Y^{-1}}_2 \\
&\leq& b^2 \norm{X^{-1}-Y^{-1}}_2. 
\end{eqnarray*}
This shows that 
\[
\norm{X^{-1}-Y^{-1}}_2 \geq \frac{1}{b^2} \norm{X-Y}_2
\]
and concludes the proof. 
\end{proof}

The function $\nabla f(X) = S - X^{-1}$ is Lipschitz continuous on any compact domain, since for $X,Y\in \mathbb{S}^p_{++}$ such that $a I \preceq X,Y \preceq b I$,
\begin{align*}
\norm{\nabla f(X) - \nabla f(Y)}_F & = \norm{X^{-1} - Y^{-1}}_F  \\
&\leq \sqrt{p} \norm{X^{-1}-Y^{-1}}_2 \\
& \leq \frac{\sqrt{p}}{a^2} \norm{X-Y}_2  \\
& \leq \frac{\sqrt{p}}{a^2} \norm{X-Y}_F.
\end{align*}

\subsection{Proof of Theorem \ref{thm:ISTA_convergence}}
We now provide the proof of Theorem 1. 

\begin{lemma}
\label{lem:ContractionBound}
Let $\Theta_{t}$ be as in Algorithm \ref{alg:ISTA} and let $\Theta^{\ast}_\rho$ be the optimal point of problem \eqref{eq:sparse_estimator}. Also, define
\begin{alignat*}{2}
b & := \max \left\{\lambda_{\mbox{max}}(\Theta_t), \lambda_{\mbox{max}}(\Theta^{\ast}_\rho)\right\}, \ \ \ a & := \min \left\{\lambda_{\mbox{min}}(\Theta_t), \lambda_{\mbox{min}}(\Theta^{\ast}_\rho)\right\}.
\end{alignat*}
Then
\begin{alignat*}{1}
\norm{\Theta_{t+1} - \Theta^{\ast}_\rho}_F & \leq \max \left\{ \abs{1-\frac{\zeta_t}{b^2}},\abs{1-\frac{\zeta_t}{a^2}} \right\} \norm{\Theta_{t} - \Theta^{\ast}_\rho}_F.
\end{alignat*}
\end{lemma}

\begin{proof}
By construction in Algorithm \ref{alg:ISTA}, 
\begin{alignat*}{1}
\Theta_{t+1} & = \eta_{\zeta_t \rho}\left((\Theta_{t} - \zeta_t(S - \Theta_{t}^{-1})\right)
\end{alignat*} 
Moreover, as $\Theta^{\ast}_\rho$ is a fixed point of the ISTA iteration \cite[Prop. 3.1]{Combettes2005}, it satisfies
\begin{alignat*}{1}
\Theta^{\ast}_\rho & = \eta_{\zeta_t \rho}\left(\Theta^{\ast}_\rho - \zeta_t(S - (\Theta^{\ast}_\rho)^{-1})\right). 
\end{alignat*}
The soft-thresholding operator $\eta_{\rho}(\cdot)$ is a proximity operator corresponding to $\rho \norm{\cdot}_{1}$. Since prox operators are non-expansive \cite[Lemma 2.2]{Combettes2005}, it follows that:
\begin{alignat*}{1}
\norm{\Theta_{t+1} - \Theta^{\ast}_\rho}_{F} & = \norm{\eta_{\zeta_t \rho}\left(\Theta_{t} - \zeta_t(S -\Theta_{t}^{-1})\right) - \eta_{\zeta_t \rho}\left(\Theta^{\ast}_\rho - \zeta_t(S -(\Theta^{\ast}_\rho)^{-1})\right) }_F \\
& \leq \norm{\Theta_{t} - \zeta_t(S - \Theta_{t}^{-1} ) - \left( \Theta^{\ast}_\rho - \zeta_t(S - (\Theta^{\ast}_\rho)^{-1}) \right) }_F  \\
& = \norm{(\Theta_{t} + \zeta_t \Theta_{t}^{-1}) - \left(\Theta^{\ast}_\rho + \zeta_t (\Theta^{\ast}_\rho)^{-1} \right) }_F 
\end{alignat*}

To bound the latter expression, recall that if $h : U \subset \mathbb{R}^{n} \rightarrow \mathbb{R}^{m}$ is a differentiable mapping, with $x,y\in U$, and $cx + (1-c)y\in U$ for all $c\in[0,1]$, then
\begin{alignat*}{1}
\norm{h(x)-h(y)} \leq \sup_{c\in[0,1]} \left\{ \norm{J_{h} \left(cx + (1-c)y \right)} \norm{x-y} \right\}
\end{alignat*} 
where $J_h(\cdot)$ is the Jacobian of $h$.  Define $h_{\gamma} : \mathbb{S}_{++}^{p} \rightarrow \mathbb{R}^{p^2}$ by 
\begin{alignat*}{1}
h_{\gamma}(X) & = \mbox{vec}(X) + \mbox{vec}(\gamma X^{-1}),
\end{alignat*}
where $\mbox{vec}(\cdot) : \mathbb{R}^{p\times p} \rightarrow \mathbb{R}^{p^2}$ is the vectorization operator defined by 
\begin{alignat*}{1}
\mbox{vec}(A) &= \left(A_{1,}, A_{2,}, \dots , A_{p,}\right)^T
\end{alignat*}
with $A_{i,}$ the $i^{th}$ row of $A$. Note that for $X\in \mathbb{S}_{++}^{p}$, 
\begin{alignat*}{1}
\frac{\partial X}{\partial X} & = I_{p^2} \quad \mbox{and} \quad \frac{\partial X^{-1}}{\partial X}  = - X^{-1} \otimes X^{-1},
\end{alignat*} 
where $\otimes$ is the Kronecker product and $I_{p^2}$ is the $p^2 \times p^2$ identity matrix. Then the Jacobian of $h_{\gamma}$ is given by:
\begin{alignat*}{1}
J_{h_{\gamma}}(X)  & = I_{p^2} - \gamma X^{-1} \otimes X^{-1}.
\end{alignat*}

Application of the mean value theorem to $h_{\zeta_t}$  over $Z_{t,c} =\mbox{vec}(c\Theta_{t} +(1-c)\Theta^{\ast}_\rho), \ c\in[0,1]$ yields
\begin{alignat*}{1}
\norm{h_{\zeta_{t}}(\Theta_{t}) - h_{\zeta_{t}}(\Theta^{\ast}_\rho)}_F & \leq \sup_{c} \left\{ \norm{I_{p^2} - \zeta_{t} Z_{t,c}^{-1} \otimes Z_{t,c}^{-1}}_2 \right\} \norm{ \mbox{vec}(\Theta_{t}) - \mbox{vec}(\Theta^{\ast}_\rho)}_2 \\
& = \sup_{c} \left\{ \norm{I_{p^2} - \zeta_{t} Z_{t,c}^{-1} \otimes Z_{t,c}^{-1}}_2 \right\} \norm{ \Theta_{t} - \Theta^{\ast}_\rho}_F  .
\end{alignat*}

Denoting the eigenvalues of $Z_{t,c}$ for given values of $t$ and $c$ as $0<\gamma_1\leq\gamma_2\leq \dots \leq \gamma_p$, the eigenvalues of $I_{p^2} - \zeta_{t} Z_{t,c}^{-1} \otimes Z_{t,c}^{-1}$ are $\left\{1 - \zeta_{t}(\gamma_i \gamma_j)^{-1}\right\}_{i,j=1}^{p}$. By Weyl's inequality,
\begin{alignat*}{2}
\gamma_p & = \lambda_{\max}(Z_{t,c}) & \leq \max \left\{ \lambda_{\max}(\Theta_{t}), \lambda_{\max}(\Theta^{\ast}_\rho)\right\} \\
\gamma_1 & = \lambda_{\min}(Z_{t,c}) & \geq \min\left\{ \lambda_{\min}(\Theta_{t}), \lambda_{\min}(\Theta^{\ast}_\rho)\right\},
\end{alignat*}
and therefore 
\begin{alignat*}{2}
\lambda_{\min}\left(I_{p^2} - \zeta_{t} Z_{t,c}^{-1} \otimes Z_{t,c}^{-1}\right) & = 1-\frac{\zeta_t}{\gamma_{1}^2} & \geq 1 - \frac{\zeta_{t}}{a^2} \\
\lambda_{\max}\left(I_{p^2} - \zeta_{t} Z_{t,c}^{-1} \otimes Z_{t,c}^{-1}\right) & = 1-\frac{\zeta_t}{\gamma_{p}^2} & \leq 1 - \frac{\zeta_{t}}{b^2}.
\end{alignat*}
Hence,
\begin{alignat*}{1}
\sup_{c} \left\{ \norm{I_{p^2} - \zeta_{t} Z_{t,c}^{-1} \otimes Z_{t,c}^{-1}}_2 \right\} & \leq \max \left\{ \abs{1-\frac{\zeta_t}{b^2}},\abs{1-\frac{\zeta_t}{a^2}} \right\}
\end{alignat*}
which completes the proof. 
\end{proof}

It follows from Lemma \ref{lem:ContractionBound} that Algorithm \ref{alg:ISTA} converges linearly if
\begin{alignat}{1}
\label{eq:EigenCondition}
s_t(\zeta_t):=\max \left\{ \abs{1-\frac{\zeta_t}{b^2}},\abs{1-\frac{\zeta_t}{a^2}} \right\} & \in (0,1), \forall t. 
\end{alignat}
Since the minimum of 
\[
s(\zeta) =\max \left\{ \abs{1-\frac{\zeta}{a^2}},\abs{1-\frac{\zeta}{b^2}} \right\}
\]
is at $\zeta = \frac{2}{a^{-2} + b^{-2}}$, Theorem 1 follows directly from Lemma 3. It now remains to show that the eigenvalues of the \texttt{G-ISTA} iterates remain bounded in eigenvalue. A more general convergence result for strongly convex functions exists in the literature; this result is stated below. 
\vspace{.5cm}
\begin{thm}
\label{thm:ISTA_convergence_general}
 Let $f$ be strongly convex with convexity constant $\mu$, and $\nabla f$ be Lipschitz continuous with constant $L$. Then for constant step size $0 < \zeta < \frac{2}{L}$, the iterates of the ISTA iteration (equation \eqref{eq:prox_iter}), $\{x_t\}_{t\geq 0}$ to minimize $f+g$ as in \eqref{eq:split_problem}, satisfy
 \begin{align*}
\norm{x_{t+1} - x^\ast}_F & \leq \max \left\{\abs{1-\zeta L}, \abs{1-\zeta \mu}  \right\}\norm{x_{t} - x^\ast}_F,
\end{align*}
which is to say that they converge linearly with rate $\max \left\{\abs{1-\zeta L}, \abs{1-\zeta \mu}  \right\}$. Furthermore,
\begin{enumerate}
\item The step size which yields an optimal worst-case contraction bound is $\zeta = \frac{2}{\mu + L}$. 
\item The optimal worst-case contraction bound corresponding to $\zeta = \frac{2}{\mu + L}$ is given by
\begin{align*}
s(\zeta) :&= \max \left\{\abs{1-\zeta L}, \abs{1-\zeta \mu}  \right\}\\
& = 1 - \frac{2}{1 + \frac{\mu}{L}}. 
\end{align*}
\end{enumerate}
\end{thm}
\begin{proof}
See \cite{Chen1997, Nesterov2007} and references therein. 
\end{proof}

\subsection{Proof of Theorem \ref{thm:eigenvalue_bounds}}

In this section, the eigenvalues of $\Theta_t, \forall t$ are bounded. To begin, the eigenvalues of $\Theta_{t+\frac{1}{2}} :=  \Theta_t - \zeta_t (S - \Theta_{t}^{-1})$ are bounded. 
\vspace{.5pc}
\begin{lemma}
\label{lemma:general_bounds}
 Let $0 < a < b $ be given positive constants and let $\zeta_t > 0$. Assume $a I \preceq \Theta_{t} \preceq b I$. Then the eigenvalues of $\Theta_{t+\frac{1}{2}} := \Theta_t - \zeta_t (S - \Theta_{t}^{-1})$ satisfy:
\begin{equation}
\label{eqn:bound_alpha}
\lambda_\textrm{min}(\Theta_{t+\frac{1}{2}}) \geq \left\{\begin{array}{cc}2 \sqrt{\zeta_t} - \zeta_t \lambda_{\max}(S)  & \textrm{ if } a \leq \sqrt{\zeta_t} \leq b \\
\min\left(a+\frac{\zeta_t}{a},b+\frac{\zeta_t}{b}\right) - \zeta_t \lambda_{\max}(S) & \textrm{ otherwise}\end{array}\right.
\end{equation}
and 
\begin{equation*}
\lambda_\textrm{max}(\Theta_{t+\frac{1}{2}}) \leq \max\left(a+\frac{\zeta_t}{a},b+\frac{\zeta_t}{b}\right) - \zeta_t \lambda_{\textrm{min}}(S).
\end{equation*}
\end{lemma}
\begin{proof}
Denoting the eigenvalue decomposition of $\Theta_{t}$ by $\Theta_{t} = U\Gamma U^T$,
\begin{alignat*}{1}
\Theta_{t+\frac{1}{2}}  & = \Theta_t - \zeta_t (S - \Theta_{t}^{-1}) \\
& = U\Gamma U^T - \zeta_t (S - U\Gamma^{-1}U^T) \\
& = U\left( \Gamma - \zeta_t ( U^T S U - \Gamma^{-1} ) \right)U^T
\end{alignat*}
Let $\Gamma = \mbox{diag}(\gamma_1, \dots, \gamma_p)$ with $\gamma_1 \leq \dots \leq \gamma_p$. By Weyl's inequality, the eigenvalues of $\Theta_{t+\frac{1}{2}}$ are bounded below by
\begin{alignat*}{1}
\lambda_{i} \left(\Theta_{t+\frac{1}{2}} \right) & \geq \gamma_i + \frac{\zeta_t}{\gamma_i} -\zeta_t \lambda_{\max}(S),
\end{alignat*} 
and bounded above by
\begin{alignat*}{1}
\lambda_{i} \left(\Theta_{t+\frac{1}{2}} \right) & \leq \gamma_i + \frac{\zeta_t}{\gamma_i} -\zeta_t \lambda_{\min}(S)
\end{alignat*}
The function $f(x) = x + \frac{\zeta_t}{x}$ over $a \leq x \leq b$ has only one extremum which is a global minimum at $x = \sqrt{\zeta_t}$. Therefore, 
\begin{equation*}
\min_{a \leq x \leq b} x + \frac{\zeta_t}{x} = \left\{\begin{array}{cc}2 \sqrt{\zeta_t} & \textrm{ if } a \leq \sqrt{\zeta_t} \leq b \\
\min\left(a+\frac{\zeta_t}{a},b+\frac{\zeta_t}{b}\right) & \textrm{ otherwise}\end{array}\right. ,
\end{equation*}
and 
\begin{equation*}
\max_{a\leq x \leq b} x + \frac{\zeta_t}{x} = \max\left(a+\frac{\zeta_t}{b},b+\frac{\zeta_t}{b}\right).
\end{equation*}

Since $a \leq \gamma_1 \leq b$,
\begin{alignat*}{1}
\lambda_{\textrm{min}}(\Theta_{t+\frac{1}{2}}) &\geq \gamma_1 + \frac{\zeta_t}{\gamma_1} -\zeta_t \lambda_{\max}(S) \\
&\geq \min_{a \leq x \leq b} \left(x+\frac{\zeta_t}{x}\right) - \zeta_t \lambda_{\max}(S)\\
&= \left\{\begin{array}{cc}2 \sqrt{\zeta_t} - \zeta_t \lambda_{\max}(S) & \textrm{ if } a\leq \sqrt{\zeta_t} \leq b \\
\min\left(a+\frac{\zeta_t}{a},b+\frac{\zeta_t}{b}\right) - \zeta_t \lambda_{\max}(S)& \textrm{ otherwise}\end{array}\right.
\end{alignat*}

Similarly, 
\begin{alignat*}{1}
\lambda_{\textrm{max}}(\Theta_{t + \frac{1}{2}}) &\leq \gamma_p + \frac{\zeta_t}{\gamma_p} - \zeta_t \lambda_\textrm{min}(S) \\
&\leq  \max_{a \leq x \leq b} \left(x + \frac{\zeta_t}{x}\right)- \zeta_t \lambda_{\textrm{min}}(S) \\
&= \max\left(a+\frac{\zeta_t}{a},b+\frac{\zeta_t}{b}\right) - \zeta_t \lambda_{\textrm{min}}(S). 
\end{alignat*}
\end{proof}

It remains to demonstrate that the soft-thresholded iterates $\Theta_{t+1}$ remain bounded in eigenvalue.
\vspace{.5pc} 

\begin{lemma}
\label{lem:equivmin}
Let $0 < a < b$ and $\zeta_t >0$. Then:
\begin{equation*}
\min\left(a + \frac{\zeta_t}{a}, b + \frac{\zeta_t}{b}\right) = a + \frac{\zeta_t}{a}
\end{equation*} 
if and only if $\zeta_t \leq a b$.
\end{lemma}
\begin{proof}
Under the stated assumptions, 
\begin{align*}
a + \frac{\zeta_t}{a} \leq b + \frac{\zeta_t}{b} &\Leftrightarrow \zeta_t\left(\frac{1}{a} - \frac{1}{b}\right) \leq b - a \\
&\Leftrightarrow \zeta_t \leq \frac{b-a}{\frac{1}{a}- \frac{1}{b}} \\
&\Leftrightarrow \zeta_t \leq a b.
\end{align*}
\end{proof}

\begin{lemma}
\label{lem:bound_soft}
Let $A$ be a symmetric $p \times p$ matrix. Then the soft-thresholded matrix $\eta_{\epsilon}(A)$ satisfies
\[
 \lambda_{\min}(A) - p \epsilon \leq \lambda_{\min}(\eta_{\epsilon}(A))
\]
In particular, $A_\epsilon$ is positive definite if $\lambda_{\min}(A) > p \epsilon$. 
\end{lemma}

\begin{proof}
Let
\[
\mathcal{A} := \{M \in \mathbb{M}_p : M_{i,j} \in \{0,1,-1\}\} . 
\]
For every $\epsilon > 0$, the matrix $A_\epsilon$ can be written as 
\[
\eta_{\epsilon}(A) = A + \epsilon_1 A_1 + \epsilon_2 A_2 + \dots + \epsilon_k A_k,
\]
for some $k \leq {p \choose 2} + p$ where $A_i \in \mathcal{A}$, $\epsilon_i > 0$ and $\sum_{i=1}^k \epsilon_i = \epsilon$. Now let 
\[
c_p := \max \{|\lambda_{\min}(M)| : M \in \mathcal{A}\}. 
\]

The constant $c_p$ is finite since $\mathcal{A}$ is a finite set. Since $-A \in \mathcal{A}$ for every $A \in \mathcal{A}$, and since $|\lambda_{\min}(-A)| = |\lambda_{\max}(A)|$, it follows that
\begin{alignat*}{1}
c_p  = \max \{|\lambda_{\max}(M)| : M \in \mathcal{A}\}. 
\end{alignat*}
Applying the Gershgorin circle theorem \citep[see, e.g.,][]{HornAndJohnson} gives $c_p \leq p$. Since $p$ is an eigenvalue of the matrix $B$ such that $B_{i,j} = 1$ for all $i,j$, it follows that $c_p = p$. 

Recursive application of Weyl's inequality gives that
\begin{eqnarray*}
\lambda_{\min}\left(\eta_{\epsilon}(A)\right) &\geq& \lambda_{\min}(A) - \epsilon |\lambda_{\max}(A_1)| - \dots - \epsilon_k |\lambda_{\max}(A_k)|\\
&\geq& \lambda_{\min}(A) - c_p \sum_{i=1}^k \epsilon_i \\
&=& \lambda_{\min}(A) - c_p \epsilon. 
\end{eqnarray*}
\end{proof}

Recall from Lemma \ref{lem:opt_bound} that the eigenvalues of the optimal solution to problem \eqref{eq:sparse_estimator} are bounded below by $\frac{1}{\|S\|_{2} + p \rho}$. The following theorem shows that $\alpha = \frac{1}{\|S\|_{2} + p \rho}$ is a valid bound to ensure that $\alpha I \preceq \Theta_{t+1}$ if $\alpha I \preceq \Theta_t$. 
\vspace{.5pc}
\begin{lemma}
\label{lem:bnd_alpha}
Let $\rho > 0$ and $\alpha = \frac{1}{\|S\|_{2} + p\rho} < b^\prime$. Assume $\alpha I\preceq \Theta_t \preceq b^\prime$ and consider 
\[
\Theta_{t+1} = \eta_{\zeta_t \rho} \left(\Theta_t - \zeta_t (S - \Theta_t^{-1})\right)
\]
Then for every $0 < \zeta_t \leq \alpha^2$, $\alpha I \preceq \Theta_{t+1}$. 
\end{lemma}

\begin{proof}
The result follows by combining Lemma \ref{lemma:general_bounds} and Lemma \ref{lem:bound_soft}. Notice first that the hypothesis $\zeta_t \leq \alpha^2$ guarantees that $\sqrt{\zeta_t} \notin [\alpha, b^\prime]$. Also, from Lemma \ref{lem:equivmin}, we have 
\begin{equation*}
\min\left(\alpha + \frac{\zeta_t}{\alpha}, b^\prime + \frac{\zeta_t}{b^\prime}\right) = \alpha + \frac{\zeta_t}{\alpha}
\end{equation*}
since $\zeta_t \leq \alpha^2 \leq \alpha b^\prime$. Hence, by Lemma \ref{lemma:general_bounds}, 
\begin{eqnarray*}
\lambda_\textrm{min} (\Theta_{t+\frac{1}{2}}) &\geq& \min\left(\alpha+\frac{\zeta_t}{\alpha},b^\prime+\frac{\zeta_t}{b^\prime}\right) - \zeta_t \lambda_{\max}(S) \\
&=& \alpha+\frac{\zeta_t}{\alpha} - \zeta_t \lambda_{\max}(S). 
\end{eqnarray*}
Now, applying Lemma \ref{lem:bound_soft} to $\Theta_{t+1} = \eta_{\zeta_t \rho}(\Theta_{t+\frac{1}{2}})$, we obtain 
\begin{eqnarray*}
\lambda_\textrm{min}(\Theta_{t+1}) &=& \lambda_\textrm{min}\left(\eta_{\zeta_t \rho}(\Theta_{t+\frac{1}{2}})\right) \\
&\geq& \lambda_\textrm{min}(\Theta_{t+\frac{1}{2}}) - p\rho\zeta_t \\
&\geq& \alpha+\frac{\zeta_t}{\alpha} - \zeta_t \lambda_{\max}(S) -  p\rho\zeta_t. 
\end{eqnarray*}
We therefore have $\alpha I \preceq \Theta_{t+1}$ whenever 
\begin{equation*}
\alpha+\frac{\zeta_t}{\alpha} - \zeta_t \lambda_{\max}(S) -  p\rho\zeta_t \geq \alpha. 
\end{equation*}
This is equivalent to 
\begin{equation*}
\zeta_t \left(\frac{1}{\alpha}-\lambda_\textrm{max}(S) - p\rho\right) \geq 0. 
\end{equation*}
Since $\zeta_t > 0$, this is equivalent to 
\begin{equation*}
\frac{1}{\alpha}-\lambda_\textrm{max}(S) - p\rho \geq 0. 
\end{equation*}
Reorganizing the terms of the previous equation, we obtain that $\alpha I \preceq \Theta_{t+1}$ if 
\begin{equation*}
\alpha \leq \frac{1}{\lambda_\textrm{max}(S)+ p\rho} = \frac{1}{\|S\|_{2}+p\rho}.
\end{equation*}
\end{proof}

It remains to show that the eigenvalues of the iterates $\Theta_t$ remain bounded above, for all $t$. 
\vspace{.5pc}
\begin{lemma}
\label{lem:bnd_beta}
Let $\alpha = \frac{1}{\|S\|_{2} + p\rho}$ and let $\zeta_t \leq \alpha^2, \forall t$. Then the G-ISTA iterates $\Theta_t$ satisfy $\Theta_t \preceq b^\prime I, \forall t$, with $b^\prime = \norm{\Theta_\rho^{\ast}}_2 + \norm{\Theta_0 - \Theta_\rho^\ast}_F$.  
\end{lemma}
\begin{proof}
By Lemma \ref{lem:bnd_alpha}, $\alpha I \preceq \Theta_t$ for every $t$. As $\alpha I \preceq \Theta^{\ast}$ (Lemma \ref{lem:opt_bound}),
\begin{equation*}
\Lambda_t^- := \min \{\lambda_\textrm{min}(\Theta_t), \lambda_\textrm{min}(\Theta^\ast_\rho)\}^2 \geq \alpha^2. 
\end{equation*}
for all $t$. Also, since $\Lambda_t^+ \geq \Lambda_t^-$ and $\zeta_t \leq \alpha^2$, 
\begin{equation*}
\max \left\{ \abs{1-\frac{\zeta_t}{b^2}},\abs{1-\frac{\zeta_t}{a^2}} \right\} \leq 1. 
\end{equation*}
Therefore, by Lemma \ref{lem:ContractionBound},
\begin{equation*}
\|\Theta_t - \Theta^\ast_\rho\|_F \leq \|\Theta_{t-1}-\Theta^\ast_\rho\|_F. 
\end{equation*}
Applying this result recursively gives
\begin{equation*}
\|\Theta_t - \Theta^\ast_\rho|_F \leq \|\Theta_0-\Theta^\ast_\rho\|_F. 
\end{equation*}
Since $\|\cdot\|_2 \leq \|\cdot\|_F$, we therefore have
\begin{equation*}
\|\Theta_t\|_2 - \|\Theta^\ast_\rho\|_2 \leq \|\Theta_t - \Theta^\ast_\rho\|_2 \leq \|\Theta_t - \Theta^\ast_\rho\|_F \leq \|\Theta_0-\Theta^\ast_\rho\|_F, 
\end{equation*}
and so, 
\begin{equation*}
\lambda_\textrm{max}(\Theta_t) = \|\Theta_t\|_2 \leq \|\Theta^\ast_\rho\|_2 + \|\Theta_0-\Theta^\ast_\rho\|_F
\end{equation*}
which completes the proof. 
\end{proof}

\subsection{Additional timing comparisons}
This section provides additional synthetic timing comparisions for $p=500$ and $p=5000$. In addition, two real datasets were investigated. The ``estrogen" dataset \cite{Pittman2004} contains $p=652$ dimensional gene expression data from $n=158$  breast cancer patients. The ``temp"  dataset \cite{Brohan06} consists of average annual temperature measurements from $p = 1732$ locations over $n = 157$ years (1850-2006). 
\begin{table}[h]
\small
\begin{center}
\begin{tabular}{|c|c|c|c|c|c|}
\cline{2-6} 
\multicolumn{1}{c|}{} & $\rho$ & $0.05$ & $0.10$ & $0.15$ & $0.20$\tabularnewline
\hline 
\textbf{problem} & \textbf{algorithm} & \textbf{time/iter} & \textbf{time/iter} & \textbf{time/iter} & \textbf{time/iter}\tabularnewline
\hline 
\hline 
 & $\mbox{nnz}(\Omega_{\rho}^{\ast})$/$\kappa(\Omega_{\rho}^{\ast})$ & $31.61\%/42.76$ & $19.61\%/18.23$ & $11.08\%/8.13$ & $5.02\%/3.06$\tabularnewline
\cline{2-6} 
$p=500$ & \texttt{glasso} & $28.34/11$ & $10.91/8$ & $7.08/7$ & $5.57/6$\tabularnewline
\cline{2-6} 
$n=100$ & \texttt{QUIC} & $8.33/23$ & $1.98/13$ & $0.96/11$ & $0.38/10$\tabularnewline
\cline{2-6} 
$\mbox{nnz}(\Omega)=3\%$ & \texttt{G-ISTA} & $\mathbf{4.44}/402$ & $\mathbf{1.14}/110$ & $\mathbf{0.30}/38$ & $\mathbf{0.14}/18$\tabularnewline
\hline 
 & $\mbox{nnz}(\Omega_{\rho}^{\ast})$/$\kappa(\Omega_{\rho}^{\ast})$ & $20.73\%/6.62$ & $3.93\%/2.44$ & $0.90\%/1.49$ & $0.13\%/1.20$\tabularnewline
\cline{2-6} 
$p=500$ & \texttt{glasso} & $7.44/6$ & $4.53/5$ & $3.45/4$ & $2.62/3$\tabularnewline
\cline{2-6} 
$n=600$ & \texttt{QUIC} & $1.08/9$ & $0.17/7$ & $\mathbf{0.06}/5$ & $0.04/5$\tabularnewline
\cline{2-6} 
$\mbox{nnz}(\Omega)=3\%$ & \texttt{G-ISTA} & $\mathbf{0.28}/31$ & $\mathbf{0.10}/13$ & $0.07/9$ & $\mathbf{0.03}/5$\tabularnewline
\hline
 & $\mbox{nnz}(\Omega_{\rho}^{\ast})$/$\kappa(\Omega_{\rho}^{\ast})$ & $31.36\%/46.83$ & $19.74\%/19.93$ & $11.65\%/8.95$ & $5.45\%/3.25$\tabularnewline
\cline{2-6} 
$p=500$ & \texttt{glasso} & $28.61/11$ & $11.27/8$ & $7.22/7$ & $5.34/6$\tabularnewline
\cline{2-6} 
$n=100$ & \texttt{QUIC} & $8.47/23$ & $2.01/13$ & $0.73/9$ & $0.22/7$\tabularnewline
\cline{2-6} 
$\mbox{nnz}(\Omega)=15\%$ & \texttt{G-ISTA} & $\mathbf{4.80}/466$ & $\mathbf{1.09}/115$ & $\mathbf{0.28}/34$ & $\mathbf{0.15}/20$\tabularnewline
\hline
 & $\mbox{nnz}(\Omega_{\rho}^{\ast})$/$\kappa(\Omega_{\rho}^{\ast})$ & $24.81\%/9.78$ & $6.36\%/2.64$ & $0.79\%/1.28$ & $0.03\%/1.08$\tabularnewline
\cline{2-6} 
$p=500$ & \texttt{glasso} & $8.52/6$ & $4.59/5$ & $3.55/4$ & $2.54/3$\tabularnewline
\cline{2-6} 
$n=600$ & \texttt{QUIC} & $1.56/10$ & $0.25/7$ & $\mathbf{0.05}/5$ & $0.03/5$\tabularnewline
\cline{2-6} 
$\mbox{nnz}(\Omega)=15\%$ & \texttt{G-ISTA} & $\mathbf{0.50}/51$ & $\mathbf{0.10}/13$ & $0.06/7$ & $\mathbf{0.02}/3$\tabularnewline
\hline 
\end{tabular}
\end{center}
\caption{Timing comparisons for $p=500$ dimensional datasets, generated as in Section \ref{subsec:synthetic_dataset}.}
\end{table}
%
%\vspace{-1cm}
%
\begin{table}[h]
\small
\begin{center}
\begin{tabular}{|c|c|c|c|c|c|}
\cline{2-6} 
\multicolumn{1}{c|}{} & $\rho$ & $0.02$ & $0.04$ & $0.06$ & $0.08$\tabularnewline
\hline 
\textbf{problem} & \textbf{algorithm} & \textbf{time/iter} & \textbf{time/iter} & \textbf{time/iter} & \textbf{time/iter}\tabularnewline
\hline 
\hline 
 & $\mbox{nnz}(\Omega_{\rho}^{\ast})$/$\kappa(\Omega_{\rho}^{\ast})$ & $26.22\%/54.47$ & $13.68\%/23.74$ & $6.36\%/8.69$ & $2.03\%/2.31$\tabularnewline
\cline{2-6} 
$p=5000$ & \texttt{glasso} & $30814.29/11$ & $12612.85/8$ & $9224.79/7$ & $6184.84/5$\tabularnewline
\cline{2-6} 
$n=1000$ & \texttt{QUIC} & $22547.70/21$ & $3725.07/11$ & $946.11/8$ & $199.48/6$\tabularnewline
\cline{2-6} 
$\mbox{nnz}(\Omega)=3\%$ & \texttt{G-ISTA} & $\mathbf{2651.43/}575$ & $\mathbf{417.20}/94$ & $\mathbf{93.33}/25$ & $\mathbf{39.05}/11$\tabularnewline
\hline 
 & $\mbox{nnz}(\Omega_{\rho}^{\ast})$/$\kappa(\Omega_{\rho}^{\ast})$ & $12.89\%/15.18$ & $3.23\%/3.73$ & $1.11\%/1.60$ & $0.16\%/1.16$\tabularnewline
\cline{2-6} 
$p=5000$ & \texttt{glasso} & $10307.26/7$ & $8725.86/7$ & $4846.58/4$ & $3587.35/3$\tabularnewline
\cline{2-6} 
$n=6000$ & \texttt{QUIC} & $3108.14/10$ & $396.60/7$ & $86.66/5$ & $\mathbf{21.56}/4$\tabularnewline
\cline{2-6} 
$\mbox{nnz}(\Omega)=3\%$ & \texttt{G-ISTA} & $\mathbf{268.28}/70$ & $\mathbf{50.17}/14$ & $\mathbf{35.67}/10$ & $28.82/8$\tabularnewline
\hline
 & $\mbox{nnz}(\Omega_{\rho}^{\ast})$/$\kappa(\Omega_{\rho}^{\ast})$ & $26.08\%/80.04$ & $13.93\%/37.12$ & $6.91\%/16.52$ & $2.47\%/3.08$\tabularnewline
\cline{2-6} 
$p=5000$ & \texttt{glasso} & $36302.86/11$ & $13413.57/8$ & $9914.41/7$ & $7408.33/6$\tabularnewline
\cline{2-6} 
$n=1000$ & \texttt{QUIC} & $22667.29/21$ & $4649.99/12$ & $1329.20/9$ & $240.25/6$\tabularnewline
\cline{2-6} 
$\mbox{nnz}(\Omega)=15\%$ & \texttt{G-ISTA} & $\mathbf{3952.85}/849$ & $\mathbf{701.57}/170$ & $\mathbf{176.11}/45$ & $\mathbf{42.46}/12$\tabularnewline
\hline
 & $\mbox{nnz}(\Omega_{\rho}^{\ast})$/$\kappa(\Omega_{\rho}^{\ast})$ & $18.65\%/27.69$ & $5.34\%/7.26$ & $0.66\%/1.41$ & $0.03\%/1.09$\tabularnewline
\cline{2-6} 
$p=5000$ & \texttt{glasso} & $13180.47/7$ & $9052.77/7$ & $4842.28/4$ & $3578.05/3$\tabularnewline
\cline{2-6} 
$n=6000$ & \texttt{QUIC} & $6600.91/12$ & $795.46/8$ & $59.03/5$ & $\mathbf{16.10}/4$\tabularnewline
\cline{2-6} 
$\mbox{nnz}(\Omega)=15\%$ & \texttt{G-ISTA} & $\mathbf{804.93}/189$ & $\mathbf{103.69}/23$ & $\mathbf{36.17}/10$ & $18.87/5$\tabularnewline
\hline 
\end{tabular}
\end{center}
\caption{Timing comparisons for $p=5000$ dimensional datasets, generated as in Section \ref{subsec:synthetic_dataset}.}
\end{table}

\begin{table}[h]
\small
\begin{center}
\begin{tabular}{|c|c|c|c|c|c|}
\cline{2-6} 
\multicolumn{1}{c|}{} & $\rho$ & $0.15$ & $0.30$ & $0.45$ & $0.60$\tabularnewline
\hline 
\textbf{problem} & \textbf{algorithm} & \textbf{time/iter} & \textbf{time/iter} & \textbf{time/iter} & \textbf{time/iter}\tabularnewline
\hline 
\hline 
 & $\mbox{nnz}(\Omega_{\rho}^{\ast})$/$\kappa(\Omega_{\rho}^{\ast})$ & $5.29\%/290.03$ & $3.39\%/88.55$ & $2.31\%/29.69$ & $1.63\%/8.96$\tabularnewline
\cline{2-6} 
$p=682$ & \texttt{glasso} & $106.18/24$ & $120.18/34$ & $110.54/35$ & $40.52/13$\tabularnewline
\cline{2-6} 
$n=158$ & \texttt{QUIC} & $\mathbf{12.36}/19$ & $\mathbf{2.71}/11$ & $\mathbf{1.08/}9$ & $\mathbf{0.54}/7$\tabularnewline
\cline{2-6} 
\emph{Dataset: estrogen} & \texttt{G-ISTA} & $43.96/2079$ & $11.99/595$ & $3.23/172$ & $1.00/53$\tabularnewline

\hline 
\cline{2-6} 
\multicolumn{1}{c|}{} & $\rho$ & $0.2$ & $0.4$ & $0.6$ & $0.8$\tabularnewline
\hline 
\textbf{problem} & \textbf{algorithm} & \textbf{time/iter} & \textbf{time/iter} & \textbf{time/iter} & \textbf{time/iter}\tabularnewline
\hline 
\hline 
 & $\mbox{nnz}(\Omega_{\rho}^{\ast})$/$\kappa(\Omega_{\rho}^{\ast})$ & $2.02\%/1075.8$ & $1.77\%/289.63$ & $1.34\%/23.02$ & $0.22\%/2.10$\tabularnewline
\cline{2-6} 
$p=1732$ & \texttt{glasso} & $1919.64/31$ & $2535.86/46$ & $1144.07/22$ & $254.14/5$\tabularnewline
\cline{2-6} 
$n=157$ & \texttt{QUIC} & $\mathbf{497.47}/18$ & $\mathbf{103.76}/13$ & $\mathbf{10.16}/8$ & $2.31/7$\tabularnewline
\cline{2-6} 
\emph{Dataset: temp} & \texttt{G-ISTA} & $1221.40/6194$ & $183.20/819$ & $30.01/159$ & $\mathbf{1.78}/10$\tabularnewline
\hline 
\end{tabular}
\end{center}
\caption{Timing comparisons for the real datasets described above.}
\label{fig:timing_real}
\end{table}

\end{document}